\documentclass[runningheads]{llncs}

\usepackage{amsmath, amssymb, amsfonts} 
\usepackage{makeidx} 

\begin{document}

\frontmatter

\mainmatter

\title{The smallest programmable machine and the hardness of analyzing it }%

\titlerunning{The smallest programmable machine} 

\author{J. Andres Montoya\inst{1}}%

\authorrunning{Montoya} 

\tocauthor{X (X)}

\institute{Universidad Nacional de Colombia, Bogota\\
\email{jamontoyaa@unal.edu.co}
}

\maketitle%

\begin{abstract}%
We investigate the computational complexity of analyzing the structural and
behavioral properties of deterministic $k$-pebble automata, which represent
a natural framework for studying minimal programmable machines. First, we
provide an explicit construction of a three-pebble automaton $\mathcal{U}%
_{3} $ capable of simulating any deterministic finite automaton (DFA) on a
given input word, establishing a link between pebble automata capabilities,
Kolmogorov complexity, and automatic complexity. Next, we restrict our
architectural framework to two-pebble programmable machines $\mathcal{U}%
_{2,N}$ simulating DFAs whose set of internal states is equal to $\left\{
1,\ldots ,N\right\} $. We analyze the decision problem \textbf{ANAL}$\left( 
\mathcal{U}_{2,5}\right) $, which asks whether a five-state program can
separate two distinct finite words. By establishing a polynomial-time
reduction from the quasi-identity checking problem for finite semigroups, we
prove that \textbf{ANAL}$\left( \mathcal{U}_{2,5}\right) $ is $\mathcal{NP}$%
-complete. Finally, we explore the complexity of separating binary strings
under the Exponential Time Hypothesis (ETH), showing how computational
hardness implies polynomial lower bounds for the size of the smallest DFAs
separating binary strings.%
\end{abstract}%

\section{Pebble automata\label{Section2}}

Let $k\geq 0$. Deterministic $k$-pebble automata are deterministic two-way
finite state automata provided with $k$ unlabeled pebbles which the finite
control can use as markers on the input tape. The study of this restricted
type of Turing machine goes back to Kreider and Ritchie \cite{Kreider}.

\begin{definition}
\label{PebbleAutomata}Let $k\geq 0.$ A deterministic $k$-pebble automaton is
a tuple%
\begin{equation*}
\mathcal{M}=\left( Q,q_{0},\Sigma ,A,\delta \right) ,
\end{equation*}%
where:

\begin{enumerate}
\item $Q$ is a finite set of states.

\item $q_{0}\in Q$ is the initial state.

\item $A\subseteq Q$ is the set of accepting states.

\item $\Sigma $ is the input alphabet.

\item The transition function 
\begin{equation*}
\delta :Q\times \Sigma \times \left( \mathcal{P}\left( \left\{ 1,\ldots
,k\right\} \right) \right) ^{2}\rightarrow Q\times \left( \left\{ 1,\ldots
,k\right\} \right) ^{2}\times \left\{ -1,0,1\right\} 
\end{equation*}%
is deterministic.
\end{enumerate}

Let $\mathcal{M}$ be a $k$-pebble automaton. Automaton $\mathcal{M}$ is a
two-way deterministic finite state automaton provided with $k$ unlabeled
pebbles. This automaton has the following capabilities:

\begin{itemize}
\item It can place its pebbles on the tape.

\item It can sense the pebbles that lie in the current cell.

\item It can pick pebbles from the set of pebbles that lie in the current
cell.
\end{itemize}

Let $\left( q,a,r,s\right) $ be a tuple such that $r+s\leq k$ holds. Suppose
that $\mathcal{M}$ is processing the input $w,$ where:

\begin{itemize}
\item $q$ is the current inner state.

\item $a$ is the character being scanned.

\item $r$ is the number of pebbles that lie in the current cell.

\item $s$ is the number of available pebbles (i.e., the number of pebbles
that lie in cell $0$).
\end{itemize}

Suppose $\delta \left( q,a,r,s\right) =\left( p,h,l,\varepsilon \right) .$
Then:

\begin{itemize}
\item $r+s=h+l.$

\item $\mathcal{M}$ changes its inner state from $q$ to $p.$

\item $\mathcal{M}$ leaves exactly $h$ pebbles in the current cell.

\item $\mathcal{M}$ moves its head in the direction indicated by $%
\varepsilon \in \left\{ -1,0,1\right\} $.
\end{itemize}
\end{definition}

\begin{definition}
We use $\mathcal{REG}_{k}$ to denote the set of languages accepted by
deterministic $k$-pebble automata.
\end{definition}

The \textit{deterministic pebble hierarchy} is the chain of inclusions:

\begin{equation*}
\mathcal{REG}_{0}\subseteq \mathcal{REG}_{1}\subseteq \mathcal{REG}%
_{2}\subseteq \cdots
\end{equation*}%
We have (see Theorem 3.2.3, page 16, \cite{Andrzej Szepietowski}):

\begin{theorem}
\label{PebblesTheorem}
$\mathcal{L} = \bigcup_{0 \leq k} \mathcal{REG}_{k}$.
\end{theorem}

\section{Simulation of DFAs using pebbles}

We construct a three-pebble automaton $\mathcal{U}_{3}$ which, on input $%
\left( \mathcal{M},w\right) $,

\begin{itemize}
\item simulates the computation of $\mathcal{M}$, on input $w$, and

\item correctly decides whether the string $w$ belongs to $L\left( \mathcal{M%
}\right) $.
\end{itemize}

Let 
\begin{equation*}
\mathcal{M}=\left( Q,\Sigma ,q_{0},A,\delta \right) .
\end{equation*}

We suppose 
\begin{equation*}
Q=\left\{ 1,...,n\right\} ,\text{ }q_{0}=1,\text{ and }A=\left\{
m_{1},...,m_{k}\right\} ,
\end{equation*}%
where $1\leq m_{1}<\cdots <m_{k}\leq n.$ We also suppose $\Sigma =\left\{
1,...,l\right\} .$

First, we assign to $\mathcal{M}$ a string $w_{\mathcal{M}}\in \left\{
0,1\right\} ^{\ast }$ that provides a full description of this automaton.
Given $i\in \mathbb{N}$, we use the symbol $1^{i}$ to denote the unary
string $\underbrace{1\cdots 1}_{i\text{ times }}$. String $w_{\mathcal{M}}$
is a suitable concatenation of the following four strings:

\begin{enumerate}
\item The string%
\begin{equation*}
w_{state}=101101110\cdots 01^{n}
\end{equation*}

\item The string 
\begin{equation*}
w_{alphabet}=10110\cdots 01^{l}.
\end{equation*}

\item The string 
\begin{equation*}
w_{accepting}=1^{m_{1}}01^{m_{2}}0\cdots 01^{m_{k}}.
\end{equation*}

\item The string $w_{\delta }$ that describes the transition function $%
\delta $. Notice that $\delta $ can be described by its table, which is a
finite list of triples of the form $\left( p,a,q\right) ,$ where $p,q\in
\left\{ 1,...,n\right\} $ and $a\in \left\{ 1,...,l\right\} .$ The triple $%
\left( p,a,q\right) $ can be described by the string $1^{p}01^{a}01^{q}.$ If 
$\delta $ is equal to 
\begin{equation*}
\left( p_{1},a_{1},q_{1}\right) ,...,\left( p_{s},a_{s},q_{s}\right) ,
\end{equation*}%
we set 
\begin{equation*}
w_{\delta
}=1^{p_{1}}01^{a_{1}}01^{q_{1}}001^{p_{2}}01^{a_{2}}01^{q_{2}}00\cdots
001^{p_{s}}01^{a_{s}}01^{q_{s}}.
\end{equation*}
\end{enumerate}

Then, we set 
\begin{equation*}
w_{\mathcal{M}}=w_{state}000w_{alphabet}000w_{accepting}000w_{\delta }.
\end{equation*}%
We observe that $\mathcal{M}$ can be fully reconstructed from the string $w_{%
\mathcal{M}}.$ We also observe that $w_{\mathcal{M}}$ can be easily computed
from $\mathcal{M}$.

\subsection{The Simulator $\mathcal{U}_{3}$}

\begin{definition}
Let $u\in \left\{ 1,\ldots ,l\right\} ^{m}$. We use $\left[ u\right] $ to
denote the binary string%
\begin{equation*}
\left[ u\right] =1^{u\left[ 1\right] }0\cdots 01^{u\left[ m\right] }.
\end{equation*}
\end{definition}

Let $\mathcal{M}$ be a DFA, and let $u$ be an input of $\mathcal{M}$. The
computation of $\mathcal{U}_{3},$ on the input%
\begin{equation*}
w_{state}000w_{alphabet}000w_{accepting}000w_{\delta }0000\left[ u\right] ,
\end{equation*}%
is divided in the following phases:

\begin{enumerate}
\item \textbf{The automaton places its pebbles.}

\begin{itemize}
\item $\mathcal{U}_{3}$ places a pebble at the left end of the substring $%
10110\cdots 01^{n}.$ This pebble is used to keep track of the inner state of 
$\mathcal{M}$. We use the term \textit{state-pebble} to designate this
pebble.

\item $\mathcal{U}_{3}$ places a second pebble at the left end of the
substring $w_{\delta }$. This pebble is used to search the table of $\delta $%
. We use the term \textit{search-pebble} to designate this pebble.

\item $\mathcal{U}_{3}$ places its third pebble at the left end of the
substring $\left[ u\right] $. This pebble is used to simulate the one-way
movement of the scanning head. We use the term \textit{scanning-pebble} to
designate this pebble.
\end{itemize}

\item \textbf{The automaton uses its pebbles.}

Let us suppose that 
\begin{equation*}
\widehat{\delta _{\mathcal{M}}}\left( 1,u\left[ 1,...,i\right] \right) =q%
\text{ and }\delta _{\mathcal{M}}\left( q,u\left[ i+1\right] \right) =p.
\end{equation*}%
Let us also suppose that the state-pebble of $\mathcal{U}_{3}$ is placed
over the substring $1^{q}$ that represents the state $q$, while the
scanning-pebble is placed over the substring that encodes the $i$-th
character of $u.$ This pebble configuration suitably represents the
configuration reached by $\mathcal{M}$ after $i$ transitions. Thus, we are
assuming that the simulation has worked well until this point. We can
suppose that the current internal state of $\mathcal{U}_{3}$ indicates that
the simulation of the $i$-th transition of $\mathcal{M}$ has just ended, and
that it is time to begin with the simulation of the $\left( i+1\right) $-th
step. Then, the automaton moves the scanning pebble to the substring that
encodes the $\left( i+1\right) $-th character of $u$. Suppose that this
character is equal to $a$. Notice that $\mathcal{U}_{3}$ can use the
search-pebble and the scanning-pebble to recognize the substrings of $%
w_{\delta }$ that encode triples of the form $\left( -,a,-\right) $. Notice
also that $\mathcal{U}_{3}$ can use the search-pebble and the state-pebble
to recognize the substrings of $w_{\delta }$ that encode triples of the form 
$\left( q,-,-\right) $. Then, $\mathcal{U}_{3}$ can use its three pebbles to
locate the substring of $w_{\delta }$ that encodes the unique triple in $%
\delta $ that looks like $\left( q,a,-\right) $. Suppose that the substring $%
1^{q}01^{a}01^{p}$ has been located. Notice that $\mathcal{U}_{3}$ can use
the search-pebble and the state-pebble to place the latter pebble over the
substring $1^{p}$ that represents state $p$. Notice that, by doing so, $%
\mathcal{U}_{3}$ has simulated a further step of the computation of $%
\mathcal{M}$ on the input $u$.
\end{enumerate}

It is important to note that we can use (basically) the same automaton $%
\mathcal{U}_{3}$ to simulate 2DFAs.

\subsection{Kolmogorov complexity and automatic complexity}

Let us fix an universal Turing machine $\mathcal{U}$ and let $w\in \Sigma
^{\ast }.$ The Kolmogorov complexity of $w$ is the size of the smallest
program for $\mathcal{U}$ that accepts $w$ and rejects any other string in $%
\Sigma ^{\left\vert w\right\vert }.$ We define the \textit{automatic
complexity} of $w$ accordingly, see \cite{Kjos}:

\begin{definition}
Let $w\in \Sigma ^{\ast }.$ The automatic complexity of $w$, denoted $%
\mathcal{K}_{\mathcal{U}_{3}}\left( w\right) ,$ is equal to%
\begin{equation*}
\left\{ 
\begin{array}{c}
\left\vert w_{\mathcal{M}}\right\vert :\mathcal{M}\text{ is a DFA that
accepts }w \\ 
\text{and rejects any other string in }\Sigma ^{\left\vert w\right\vert }%
\end{array}%
\right\} .
\end{equation*}
\end{definition}

\begin{question}
Consider the problem:\medskip

\begin{itemize}
\item \textit{Input: }$\left( w,k\right) $\textit{, where }$w$\textit{\ is a
finite word and }$k$\textit{\ is a positive integer.}

\item \textit{Problem: decide whether }$\mathcal{K}_{\mathcal{U}_{3}}\left(
w\right) \geq k$\textit{.}\medskip
\end{itemize}

What is the complexity of this problem?
\end{question}

\section{The smallest programmable machine}

Let $Q_{N}=\left\{ 1,\ldots ,N\right\} $. We can use just two pebbles (and
the architecture of $\mathcal{U}_{3}$) to simulate all the DFAs with at most 
$N$ states. We use $\mathcal{U}_{2,N}$ to denote the corresponding
two-pebble automaton. It is worth to emphasize that $\mathcal{U}_{2,N}$ is a
Turing machine with the ability of executing infinitely many different
programs (DFAs).

\begin{example}
$\mathcal{U}_{2,1}$ only accepts trivial languages. $\mathcal{U}_{2,2}$
accepts all the alphabet testable languages (i.e., languages of the form 
$\bigcup_{a \in \Gamma} \Sigma^{*}a\Sigma^{*}$). What
languages are accepted by $\mathcal{U}_{2,3}$ and $\mathcal{U}_{2,4}$ (by
the programs that can be executed by these two programmable machines)?
\end{example}

It could be argued that the set $\left\{ \mathcal{U}_{2,N}:N\geq 2\right\} $
contains the smallest (simplest) programmable machines. We show that, for
all $N\geq 5$, the programmable machine $\mathcal{U}_{2,N}$ is hard to
analyze.

\begin{definition}
Let $u,v\in \Sigma ^{\ast }.$ We define:%
\begin{equation*}
Sep\left( u,v\right) =\min \left\{ 
\begin{array}{c}
k:\text{there exists a DFA with }k\text{ states} \\ 
\text{that separates }u\text{ from }v%
\end{array}%
\right\} .
\end{equation*}
\end{definition}

We have:%
\begin{equation*}
Sep\left( u,v\right) =\min \left\{ 
\begin{array}{c}
N:\text{there exists a program for }\mathcal{U}_{2,N} \\ 
\text{that separates }u\text{ from }v%
\end{array}%
\right\} .
\end{equation*}

Consider the following problem:\medskip

\begin{problem}
We use \textbf{ANAL}$\left( \mathcal{U}_{2,5}\right) $ to denote the
following algorithmic problem. 
\end{problem}

\textit{Input: }$\left( w,u\right) $\textit{, where }$w$\textit{\ and }$u$ 
\textit{are finite words.}

\textit{Problem: decide whether there exists a program for }$\mathcal{U}%
_{2,5}$ \textit{that separates} $u$ \textit{from} $w$\textit{.}\medskip 

We use \textbf{ANAL}$\left( \mathcal{U}_{2,5}\right) $ to denote the above
problem. We show that this problem is NP-complete.

\subsection{Identities of finite semigroups}

\begin{definition}
Let $k>0$. A transformation semigroup over the set $\left\{ 1,...,k\right\} $
is a subset of $\left\{ 1,...,k\right\} ^{\left\{ 1,...,k\right\} }$ that is
closed under composition.
\end{definition}

\begin{example}
The set $\left\{ 1,...,k\right\} ^{\left\{ 1,...,k\right\} }$ is itself a
transformation semigroup and it is called the \textit{full transformation
semigroup} of order $k.$ We use the symbol $\mathcal{F}_{k}$ to denote this
semigroup.
\end{example}

\begin{definition}
Let $\Sigma $ be a finite alphabet. An \textit{interpretation} of $\Sigma $
in the semigroup $\mathcal{F}_{k}$ is a function $\alpha :\Sigma \rightarrow 
\mathcal{F}_{k}$ that assigns to each letter in $\Sigma $ an element of $%
\mathcal{F}_{k}$. Given $w\in \Sigma ^{\ast }$ and given $\alpha $ we
interpret $w$ as the function%
\begin{equation*}
f_{w,\alpha }=\alpha \left( w\left[ n\right] \right) \circ \cdots \circ
\alpha \left( w\left[ 1\right] \right) .
\end{equation*}
\end{definition}

Let $\alpha :\Sigma \rightarrow \mathcal{F}_{k}$ be an interpretation. We
have:

\begin{itemize}
\item $\alpha $ determines a DFA $\mathcal{M}_{\alpha }$ whose input
alphabet equals $\Sigma $ and whose set of internal states coincides with
the set $\left\{ 1,...,k\right\} ,$ and

\item $\alpha $ determines a semigroup of $\mathcal{F}_{k}$, the semigroup
spanned by the image of $\alpha .$
\end{itemize}

\begin{definition}
We say that $\left( u,v\right) \in \Sigma ^{\ast }$ is an identity for $%
\mathcal{F}_{k}$ if and only if for all interpretation $\alpha :\Sigma
\rightarrow \mathcal{F}_{k}$ the equality $f_{u,\alpha }=f_{v,\alpha }$
holds.
\end{definition}

Suppose that $\left( u,v\right) $ is an identity for $\mathcal{F}_{k}.$ Then:

\begin{enumerate}
\item For all $i\leq k,$ the pair $\left( u,v\right) $ is an identity for $%
\mathcal{F}_{i}.$

\item If $u\neq v,$ there exists $k$ such that $\left( u,v\right) $ is not
an identity for $\mathcal{F}_{k}.$
\end{enumerate}

\begin{definition}
We define:%
\begin{equation*}
I\left( u,v\right) =\max \left\{ k:\left( u,v\right) \text{ is an identity
for }\mathcal{T}_{k}\right\} .
\end{equation*}
\end{definition}

Consider the following problem:\medskip 

\begin{problem}
\textbf{IDENT}

\begin{itemize}
\item \textit{Input: }$\left( u,v,k\right) ,$ \textit{where }$u,v$\textit{\
are strings and }$k$\textit{\ is a positive integer.}

\item \textit{Problem: decide whether }$I\left( u,v\right) \geq k$\textit{.}%
\smallskip
\end{itemize}
\end{problem}

IDENT is co-$\mathcal{NP}$ hard; see \cite{Volkov}. Furthermore, the
following subproblem is also co-$\mathcal{NP}$ hard:%
\begin{equation*}
\text{\textbf{IDENT}}_{5}=\left\{ \left( u,v\right) :\left( u,v,5\right) \in 
\text{\textbf{IDENT}}\right\} .
\end{equation*}%
We use this to prove that \textbf{ANAL}$\left( \mathcal{U}_{2,5}\right) $ is 
$\mathcal{NP}$-hard.

\begin{lemma}
$I\left( u,v\right) +1=Sep\left( u,v\right) .$
\end{lemma}

\begin{proof}
Let $\Sigma $ be the alphabet of $u$ and $v.$ Suppose $\Sigma =\left\{
a_{1},...,a_{n}\right\} .$ Let 
\begin{equation*}
f_{a_{1}},...,f_{a_{n}}:\left\{ 1,...,N\right\} \rightarrow \left\{
1,...,N\right\} .
\end{equation*}%
These $n$ functions determine a DFA with $N$ states and input alphabet $%
\Sigma .$ We use the symbol $\mathcal{M}_{f_{a_{1}},...,f_{a_{n}}}$ to
denote this automaton. Automaton $\mathcal{M}_{f_{a_{1}},...,f_{a_{n}}}$ is
equal to $\left( \left\{ 1,...,N\right\} ,\Sigma ,\delta \right) ,$ where
the transition function $\delta $ is defined by the equation 
\begin{equation*}
\delta \left( q,a_{i}\right) =f_{a_{i}}\left( q\right) 
\end{equation*}%
Notice that $\mathcal{M}_{f_{a_{1}},...,f_{a_{n}}}$ separates $\left(
u,v\right) $ if and only if the equality $f_{u,\alpha }=f_{v,\alpha }$ does
not hold, where $\alpha $ is the interpretation given by the equalities: 
\begin{equation*}
\alpha \left( a_{i}\right) =f_{a_{i}},\text{ for all }a_{i}\in \Sigma .
\end{equation*}%
We get the first implication, that is: if $\left( u,v\right) $ is an
identity for $\mathcal{F}_{N}$ the pair $\left( u,v\right) $ cannot be
separated with $N$ states.

Let $\mathcal{M}$ be a $N$-state automaton with input alphabet $\left\{
a_{1},...,a_{n}\right\} .$ Automaton $\mathcal{M}$ determines an
interpretation $\alpha _{\mathcal{M}}:\left\{ a_{1},...,a_{n}\right\}
\rightarrow \mathcal{F}_{N}$ which is given by:%
\begin{equation*}
\alpha _{\mathcal{M}}\left( a_{i}\right) \left( q\right) =\delta _{\mathcal{M%
}}\left( q,a_{i}\right) 
\end{equation*}%
Notice that $\mathcal{M}$ separates $\left( u,v\right) $ if and only if the
equality $f_{u,\alpha _{\mathcal{M}}}=f_{v,\alpha _{\mathcal{M}}}$ does not
hold. We get the second implication. The lemma follows.
\end{proof}

We obtain:

\begin{theorem}
\label{NP-hard}\textbf{ANAL}$\left( \mathcal{U}_{2,5}\right) $ is $\mathcal{%
NP}$-complete.
\end{theorem}

\begin{proof}
The problem \textbf{ANAL}$\left( \mathcal{U}_{2,5}\right) $ belongs to $%
\mathcal{NP}$. We show that this problem is $\mathcal{NP}$-hard.

Let $\left( u,v\right) $ be an instance of \textbf{IDENT}$_{5}$ (\textbf{ANAL%
}$\left( \mathcal{U}_{2,5}\right) $). This pair is an identity of $\mathcal{F%
}_{5}$ if and only if it cannot be separated with five states. We get that 
\textbf{ANAL}$\left( \mathcal{U}_{2,5}\right) $ coincides with $co$-\textbf{%
IDENT}$_{5}$. The theorem follows.
\end{proof}

We also obtain the corollary below. First, consider the following
algorithmic problem:\medskip 

\begin{problem}
\textbf{SEP}

\begin{itemize}
\item \textit{Input: }$\left( u,v,k\right) ,$ \textit{where }$u,v$\textit{\
are strings and }$k$\textit{\ is a positive integer.}

\item \textit{Problem: decide whether }$Sep\left( u,v\right) \geq k$\textit{.%
}\smallskip
\end{itemize}
\end{problem}

We have:

\begin{corollary}
\label{hardness}The problem \textbf{SEP }is $\mathcal{NP}$-complete.
\end{corollary}

\subsection{Separating binary strings}

The proofs of Theorem \ref{NP-hard} and Corollary \ref{hardness} make
extensive use of the following fact: the definition of \textbf{SEP} did not
impose an upper bound on the size of the alphabets used by the input
strings. We would like to focus on the case $\Sigma =\left\{ 0,1\right\} .$

\begin{definition}
Let $n\geq 1$. We define:%
\begin{equation*}
Sep_{2}\left( n\right) =\max \left\{ Sep\left( u,v\right) :u,v\in \left\{
0,1\right\} ^{n}\right\} .
\end{equation*}
\end{definition}

We have, see \cite{Goralcik} and \cite{Zacharias}:

\begin{enumerate}
\item $Sep_{2}\left( n\right) \in \Omega \left( \log \left( n\right) \right)
.$

\item $Sep_{2}\left( n\right) \in O\left( n^{\frac{1}{3}}\cdot \log
^{7}\left( n\right) \right) $
\end{enumerate}

The gap between upper and lower bounds is very large. We conjecture that,
for all $m\in \mathbb{N}$, the function $Sep_{2}\left( n\right) $ does not
belong to $O\left( \log ^{m}\left( n\right) \right) $. Let us observe that
computational complexity can shed some light on this conjecture. Consider
the following problem:\medskip

\begin{problem}
\textbf{SEP}$_{2}$

\begin{itemize}
\item \textit{Input: }$\left( u,v,k\right) ,$ \textit{where }$u,v$\textit{\
are binary strings and }$k$\textit{\ is a positive integer.}

\item \textit{Problem: decide whether }$Sep\left( u,v\right) \geq k$\textit{.%
}\smallskip
\end{itemize}
\end{problem}

We have:

\begin{proposition}
Assume the exponential time hypothesis (i.e., assume that SAT cannot be
solved in subexponential time). Then, if \textbf{SEP}$_{2}$ is $\mathcal{NP}$%
-complete the function $Sep_{2}\left( n\right) $ does not belong to $O\left(
\log ^{m}\left( n\right) \right) .$
\end{proposition}

\textbf{Acknowledgement. }The author thanks Centro de Investigaciones
Antonio S\'{a}nchez de Cozar, vereda los Egidos, San Gil, Colombia. Most of
this work was written while the author was visiting this charming research
center in the Santander mountains.

\end{document}